\documentstyle[11pt,epsfig]{article}

\newcommand {\ignore} [1] {}

\newtheorem{theorem}{Theorem}[section]
\newtheorem{lemma}[theorem]{Lemma}
\newtheorem{fact}[theorem]{Fact}
\newtheorem{corollary}[theorem]{Corollary}
\newtheorem{definition}{Definition}[section]
\newtheorem{proposition}[theorem]{Proposition}

\newtheorem{claim}[theorem]{Claim}

\newenvironment{proof}{\noindent{\bf Proof:\/}}{\hfill $\Box$\vskip 0.1in}

\begin{document}

\date{}

\title{Approximating minimum-cost edge-covers of crossing biset-families\thanks{Part of this paper 
appeared in the preliminary version \cite{N}.}}

\author{
Zeev Nutov \\
\small The Open University of Israel \\
\small {\tt nutov@openu.ac.il}
}

\maketitle

\begin{abstract}
An ordered pair $\hat{S}=(S,S^+)$ of subsets of a groundset $V$ is called a {\em biset} if $S \subseteq S^+$;
$(V \setminus S^+,V \setminus S)$ is the {\em co-biset} of $\hat{S}$.
Two bisets $\hat{X},\hat{Y}$ 
{\em intersect} if $X \cap Y \neq \emptyset$ and 
{\em cross} if both $X \cap Y \neq \emptyset$ and $X^+ \cup Y^+ \neq V$.
The intersection and the union of two bisets $\hat{X},\hat{Y}$ is defined by
$\hat{X} \cap \hat{Y} = (X \cap Y, X^+ \cap Y^+)$ and $\hat{X} \cup \hat{Y} = (X \cup Y,X^+ \cup Y^+)$.
A biset-family ${\cal F}$ is {\em crossing (intersecting)} if 
$\hat{X} \cap \hat{Y}, \hat{X} \cup \hat{Y} \in {\cal F}$ 
for any $\hat{X},\hat{Y} \in {\cal F}$ that cross (intersect). 
A directed edge covers a biset $\hat{S}$ if  it goes from $S$ to $V \setminus S^+$.
We consider the problem of covering a crossing biset-family ${\cal F}$ by a minimum-cost set of directed edges.
While for intersecting ${\cal F}$, a standard primal-dual algorithm computes an optimal solution,
the approximability of the case of crossing ${\cal F}$ is not
yet understood, as it includes several NP-hard problems, for which a poly-logarithmic approximation
was discovered only recently or is not known.
Let us say that a biset-family ${\cal F}$ is {\em $k$-regular} if 
$\hat{X} \cap \hat{Y}, \hat{X} \cup \hat{Y} \in {\cal F}$ for any $\hat{X},\hat{Y} \in {\cal F}$ 
with $|V \setminus (X \cup Y)| \geq k+1$ that intersect.
In this paper we obtain an $O(\log |V|)$-approximation algorithm for arbitrary crossing ${\cal F}$;
if in addition both ${\cal F}$ and the family of co-bisets of ${\cal F}$ are $k$-regular, 
our ratios are:
$O\left(\log \frac{|V|}{|V|-k} \right)$ if $|S^+ \setminus S|=k$ for all $\hat{S} \in {\cal F}$,
and 
$O\left(\frac{|V|}{|V|-k} \log \frac{|V|}{|V|-k}\right)$ 
if $|S^+ \setminus S| \leq k$ for all $\hat{S} \in {\cal F}$.
Using these generic algorithms, we derive for some network design problems 
the following approximation ratios:
$O\left(\log k \cdot \log\frac{n}{n-k}\right)$ for {\sf $k$-Connected Subgraph},  and 
$O(\log k) \cdot \min\{\frac{n}{n-k} \log \frac{n}{n-k},\log k\}$ 
for {\sf Subset $k$-Connected Subgraph} when all edges with positive cost have their endnodes in the subset.  
\end{abstract}

\section{Introduction}

\subsection{Problem definition and main results}

Following \cite{F-R}, an ordered pair $\hat{S}=(S,S^+)$ of subsets of a groundset $V$ is called a {\em biset} 
if $S \subseteq S^+$; $S$ is the {\em inner part} and $S^+$ is the {\em outer part} of $\hat{S}$,
and $\Gamma(\hat{S})=S^+ \setminus S$ is the {\em boundary} of $\hat{S}$. 
The {\em co-biset} of $\hat{S}$ is the biset $(V \setminus S^+,V \setminus S)$.
Any set $S$ can be considered as a biset $\hat{S}=(S,S)$ with $\Gamma(\hat{S})=\emptyset$. 

\begin{definition} \label{d:1}
Two bisets $\hat{X},\hat{Y}$ on $V$ 
{\em intersect} if $X \cap Y \neq \emptyset$ and are {\em disjoint} otherwise;
$\hat{X},\hat{Y}$ {\em cross} if $X \cap Y \neq \emptyset$ and $X^+ \cup Y^+ \neq V$.
The intersection and the union of bisets $\hat{X},\hat{Y}$ is defined by
$\hat{X} \cap \hat{Y} = (X \cap Y, X^+ \cap Y^+)$ and $\hat{X} \cup \hat{Y} = (X \cup Y,X^+ \cup Y^+)$.
We say that a biset-family ${\cal F}$ is: 
\begin{itemize}
\item
{\em crossing (intersecting)} if $\hat{X} \cap \hat{Y}, \hat{X} \cup \hat{Y} \in {\cal F}$ 
for any $\hat{X},\hat{Y} \in {\cal F}$ that cross (intersect).
\item
{\em $k$-regular} if $\hat{X} \cap \hat{Y}, \hat{X} \cup \hat{Y} \in {\cal F}$ 
for any $\hat{X},\hat{Y} \in {\cal F}$ with $|V \setminus (X \cup Y)| \geq k+1$ that intersect.
\end{itemize}
\end{definition}

A biset $\hat{S}$ is {\em proper} if $S,V \setminus S^+$ are both nonempty.
All biset-families in this paper are assumed to contain only proper bisets.
A directed edge $e$ leaves/covers a (proper) biset $\hat{S}$ if it goes from $S$ to $V \setminus S^+$.
An edge-set/graph $J$ is an {\em edge-cover} of ${\cal F}$ 
if every biset in ${\cal F}$ is covered by some edge in $J$.
We consider the following generic problem.

\vspace{0.1cm}

\begin{center} 
\fbox{
\begin{minipage}{0.965\textwidth}
\noindent
{\sf Biset-Family Edge-Cover}  \\ 
{\em Instance:} \ A directed graph $G=(V,E)$ with edge-costs $\{c_e:e \in E\}$ and a biset-family ${\cal F}$ on $V$. \\
{\em Objective:}  Find a minimum-cost edge-cover $J \subseteq E$ of ${\cal F}$.
\end{minipage}
}
\end{center}

\vspace{0.1cm}

Given two bisets $\hat{X},\hat{Y}$ we write $\hat{X} \subseteq \hat{Y}$ and 
say that $\hat{Y}$ contains $\hat{X}$ if $X \subseteq Y$
or if $X=Y$ and $X^+ \subseteq Y^+$; similarly, $\hat{X} \subset \hat{Y}$ and $\hat{Y}$ 
properly contains $\hat{X}$ if $X \subset Y$ or if $X=Y$ and $X^+ \subset Y^+$.

\begin{definition}
A biset $\hat{C} \in {\cal F}$ is a {\em core} of a biset-family ${\cal F}$,
or an {\em ${\cal F}$-core} for short, 
if $\hat{C}$ contains no biset in ${\cal F} \setminus \{\hat{C}\}$, namely,
if $\hat{C}$ is an inclusion-minimal member of ${\cal F}$.
Let ${\cal C}({\cal F})$ denote the family of ${\cal F}$-cores and let 
$\nu({\cal F})=|{\cal C}({\cal F})|$ denote the number of ${\cal F}$-cores.
\end{definition}

In the {\sf Biset-Family Edge-Cover} problem, ${\cal F}$ may not be given explicitly,
and a polynomial in $n=|V|$ implementation of our algorithms requires that certain queries related 
to ${\cal F}$ can be answered in polynomial time.
Given an edge set $J$ on $V$, the {\em residual family} ${\cal F}^J$
of ${\cal F}$ consists of all members of ${\cal F}$ that are uncovered by the edges of $J$.
It is known that if ${\cal F}$ is crossing or intersecting, so is ${\cal F}^J$, for any $J$.
The {\em co-family} of ${\cal F}$ is the biset-family 
$\{(V \setminus S^+, V\setminus S):(S,S^+) \in {\cal F}\}$ of co-bisets of the bisets in ${\cal F}$.
It is easy to see that ${\cal F}$ is crossing if, and only if, its co-family is crossing,
and that $J$ covers ${\cal F}$ if, and only if, the reverse edge-set of $J$ covers the co-family 
of ${\cal F}$. We assume that for any edge set $J$ on $V$ and any $u,v \in V$ we are able to compute 
in polynomial time the cores of the biset-family 
${\cal F}(u,v)=\{\hat{S} \in {\cal F}^J: u \in S, v \in V \setminus S^+\}$ 
and also the cores of its co-family, or to determine that ${\cal F}(u,v)$ is empty.
In specific graph problems we consider, this can be implemented in polynomial time using the 
Ford-Fulkerson Max-Flow Min-Cut Algorithm; we omit the somewhat standard implementation details.

For intersecting ${\cal F}$, {\sf Biset-Family Edge-Cover} can be solved in polynomial time by a 
standard primal-dual algorithm; in fact, even a more general problem of covering an intersecting supermodular 
biset-function by a digraph can also be solved in polynomial time \cite{F-R}. 
However, the case of crossing ${\cal F}$ includes the min-cost {\sf $k$-Connectivity Augmentation} problem
which is NP-hard, and its approxi\-mability is not yet understood.
Given a biset-family ${\cal F}$ let $\gamma({\cal F})=\max\limits_{\hat{S} \in {\cal F}}|\Gamma(\hat{S})|$.
It is known that any crossing set-family ${\cal F}$ 
(namely, a crossing biset-family ${\cal F}$ with $\gamma({\cal F})=0$) 
is decomposable into two intersecting families 
${\cal F}^{out}=\{\hat{S} \in {\cal F}:s \in S\}$ and 
${\cal F}^{in} =\{\hat{S} \in {\cal F}:s \in V \setminus S^+\}$, where $s \in V$ is arbitrary,
such that an edge-set $J$ covers ${\cal F}$ if, and only if, $J$ covers ${\cal F}^{out}$ 
and the reverse edge-set of $J$ covers ${\cal F}^{in }$.
This implies ratio $2$ for {\sf Biset-Family Edge-Cover} for crossing ${\cal F}$
with $\gamma({\cal F})=0$.
In a similar way we can decompose any crossing ${\cal F}$ into 
$2(\gamma({\cal F})+1)$ intersecting biset-families. 
This implies ratio $2(\gamma({\cal F})+1)$ for {\sf Biset-Family Edge-Cover} with crossing ${\cal F}$. 
Using ideas from \cite{RW,KN2,FL,N,N-subs}, we give approximation algorithms with logarithmic ratios.

For an edge-set or a graph $J$ and a biset $\hat{S}$ on $V$
let $\delta_J(\hat{S})$ denote the set of edges in $J$ covering $\hat{S}$. 
Let $\tau({\cal S})$ denote the optimal value of an LP-relaxation for covering
a biset-family ${\cal S}$, namely,
\[ \displaystyle
\begin{array} {llll} \mbox{\bf (P)}
& \tau({\cal S}) = & \min          & \ \ \displaystyle \sum_{e \in E} c_e x_e   \\
&                  & \ \mbox{s.t.} & \displaystyle \sum_{e \in \delta_E(\hat{S})} x_e \geq 1 
                                     \ \ \ \ \ \ \forall \hat{S} \in {\cal S}   \\
&                  &               & \ \ \ x_e \geq 0 \ \ \ \ \ \ \ \ \ \ \ \ \ \forall e \in E
\end{array}
\]

Let $H(n)=\sum_{i=1}^n (1/i)$ denote the $n$th harmonic number.
Our main result is the following.

\vspace{0.2cm}

\begin{theorem} \label{t:FL} 
{\sf Biset-Family Edge-Cover} with crossing ${\cal F}$ (and a directed graph $G$) 
admits a polynomial time algorithm that computes a solution of cost at most
$\rho \cdot \tau({\cal F})$ where: 
\begin{itemize}
\item[{\em (i)}]
$\rho= O(\log \nu({\cal F}))=O(\log n)$ for arbitrary crossing ${\cal F}$. 
\item[{\em (ii)}]
$\rho=O\left(\log \min\left\{\frac{n}{n-k},n-k\right\}\right)$
if both ${\cal F}$ and the co-family of ${\cal F}$ are $k$-regular 
and if $|\Gamma(\hat{S})|=k$ for all $\hat{S} \in {\cal F}$.
\item[{\em (iii)}]
$\rho=O\left(\frac{n}{n-k} \log \frac{n}{n-k}\right)$ if both ${\cal F}$ 
and the co-family of ${\cal F}$ are $k$-regular and if $\gamma({\cal F}) \leq k$.
\end{itemize}
\end{theorem}

We note that Theorem~\ref{t:FL} easily extends to the undirected case,
with a loss of a factor of $2$ in the approximation ratio.

\subsection{Related work and applications}

A directed/undirected graph is {\em $k$-connected} if there are
$k$ internally-disjoint paths from every its node to the other.
A fundamental problem in network design is the following: 

\vspace{0.1cm}

\begin{center} 
\fbox{
\begin{minipage}{0.965\textwidth}
\noindent
{\sf $k$-Connected Subgraph} \\
{\em Instance:}
\ A graph $G=(V,E)$ with edge-costs $\{c_e:e \in E\}$ and an integer $k$. \\
{\em Objective:}
Find a minimum cost $k$-connected spanning subgraph of $G$.
\end{minipage}
}
\end{center}

\vspace{0.1cm}

We refer the reader to \cite{N,FL,KN-sur} for a history of the problem.
Let the {\sf $\ell$-Connectivity Augmentation} problem be the restriction of 
{\sf $k$-Connected Subgraph} to instances in which $G$ contains an $\ell$-connected spanning 
subgraph $G_0$ of cost zero, and we seek to increase at minimum cost
the connectivity of $G_0$ from $\ell=k-1$ to $\ell+1=k$. 
The following statement is known, and its parts were implicitly proved 
in \cite{FJ} and \cite{RW}, see also \cite{KN2}.

\begin{proposition} \label{p:RW}
{\sf $\ell$-Connectivity Augmentation} is a particular case of {\sf Biset-Family Edge-Cover} with crossing 
${\cal F}$, such that $|\Gamma(\hat{S})|=\ell$ for all $\hat{S} \in {\cal F}$, 
and such that both ${\cal F}$ and the co-family of ${\cal F}$ are $\ell$-regular.
Furthermore, if the latter problem admits a polynomial time algorithm that computes a solution
of cost $\leq \alpha(n,\ell) \cdot \tau({\cal F})$, then {\sf $k$-Connected Subgraph}
admits a polynomial time algorithm that computes a solution 
of cost $\leq {\sf opt}_k \cdot \sum_{\ell=1}^k \frac{\alpha(n,\ell)}{k-\ell+1}$.
In particular, if $\alpha(n,\ell)$ is increasing in $\ell$ then the cost of the solution
computed $\leq {\sf opt}_k \cdot \alpha(n,k) \cdot H(k)$, where 
${\sf opt}_k=\min\{x(E): 
  x(\delta_E(\hat{S})) \geq k-|\Gamma(\hat{S})| \ \forall \mbox{ proper biset } \hat{S} \mbox{ on } V\}$
is the optimal value of a natural LP-relaxation for the problem.
\end{proposition}

Thus part~(ii) of Theorem~\ref{t:FL} implies the following result from \cite{N}.

\begin{corollary} [\cite{N}] \label{c:aug}
{\sf $k$-Connectivity Augmentation} admits a polynomial time algorithm 
that computes a solution of cost 
$O\left(\log \frac{n}{n-k} \right) \cdot {\sf opt}_k$;
the approximation ratio is $O(1)$, unless $k=n-o(n)$.
The problem also admits a polynomial time algorithm that computes a solution of cost
$O(\log(n-k))\cdot {\sf opt}_k$.
{\sf $k$-Connected Subgraph} admits an 
$O\left(\log k \cdot \log \frac{n}{n-k} \right)$-approximation algorithm;
the ratio is $O(\log k)$, unless $k=n-o(n)$.
The problem also admits a polynomial time algorithm that computes a solution of cost
$O\left(\sum_{\ell=1}^k \frac{\log (n-\ell)}{k-\ell+1}\right) \cdot {\sf opt}_k$.
\end{corollary}

Now let us consider the following known generalization of the {\sf $k$-Connected Subgraph} problem.
Let us say that a subset $T$ of nodes of a directed/undirected graph is $k$-connected if in the graph 
there are $k$ internally-disjoint paths from every node in $T$ to any other node in $T$.

\vspace{0.1cm}

\begin{center} 
\fbox{
\begin{minipage}{0.965\textwidth}
\noindent
{\sf Subset $k$-Connected Subgraph} \\
{\em Instance:}
\ A graph $G=(V,E)$ with edge-costs $\{c_e:e \in E\}$, $T \subseteq V$, and an integer $k$. \\
{\em Objective:}
Find a minimum cost subgraph of $G$ in which $T$ is $k$-connected.
\end{minipage}
}
\end{center}

\vspace{0.1cm}

Note that the {\sf $k$-Connected Subgraph} problem is a particular case of 
{\sf Subset $k$-Connected Subgraph} when $T=V$.
Let the {\sf Subset $\ell$-Connectivity Augmentation} problem be the restriction of 
{\sf Subset $k$-Connected Subgraph} to instances in which $G$ contains a 
subgraph $G_0$ of cost zero such that $T$ is $\ell$-connected in $G_0$,
and we seek to increase at minimum cost
the connectivity of $T$ from $\ell=k-1$ to $\ell+1=k$. 
When the costs are arbitrary, {\sf Subset $k$-Connected Subgraph} is unlikely to admit 
a polylogarithmic approximation \cite{KKL} (see also \cite{LN} for a simpler proof).
The currently best known ratio for this problem for $|T|>k$ is
$b(k+\rho)+\frac{|T|^2}{{(|T|-k)}^2} O\left(\log \frac{|T|}{|T|-k}\right)$,
where $b=1$ for undirected graphs and $b=2$ for directed graphs, and
$\rho$ is the ratio for the rooted version of the problem \cite{L-subs,N-subs}; 
currently, $\rho=\min\{\tilde{O}(k),|T|\}$ for undirected graphs \cite{N-focs}, and 
$\rho=|T|$ for directed graphs. For $|T| \leq k$ the best ratio is $\frac{b}{2}|T|(|T|-1)$.
We consider the version of the problem when every edge with positive cost has its both endnodes
in $T$. Then a similar statement to the one in Proposition~\ref{p:RW} applies, 
except that ${\cal F}$ is a biset-family on $T$ and $|\Gamma(\hat{S})| \leq \ell$ for all $\hat{S} \in {\cal F}$.
Furthermore, when $|T|>k$, then by applying $b$ times an approximation algorithm for
the rooted version of the problem, we can reduce the number of cores to $O(k^2)$;
such a procedure is described in \cite{L-subs,N-subs}.
The rooted version when every edge has its tail in $T$ admits a polynomial time algorithm \cite{F-R}.
Thus parts (i) and (iii) of Theorem~\ref{t:FL} imply the following.

\begin{corollary} \label{t:aug'}
In the case when every edge with positive cost has its both endnodes in $T$,
{\sf Subset $k$-Connectivity Augmentation} admit a polynomial time algorithm that computes a solution 
of cost $\rho \cdot {\sf opt}_k$, where 
$\rho=
O\left(\min\left\{\frac{|T|}{\max\{|T|-k,1\}} \log \frac{|T|}{\max\{|T|-k,1\}},\log \min\{k,|T|\}\right\}\right)$
and ${\sf opt}_k=\min\{x(E): x(\delta_E(S)) \geq k-|\Gamma(\hat{S})| \ 
\forall \mbox{ biset } \hat{S} \mbox{ on } V \mbox{ with } 
S \cap T,T \setminus S^+ \neq \emptyset\}$.
{\sf Subset $k$-Connected Subgraph} admits an $O(\rho \log k)$-approximation algorithm.
\end{corollary}

Cheriyan and Laekhanukit \cite{ChL} considered the following directed edge-connectivity problem,
that gene\-ra\-lizes the {\sf Subset $k$-Connected Subgraph} problem.
Given two disjoin subsets $S,T$ in a graph $G$, we say that $G$ is $k$-edge-outconnected from $S$ to $T$,
or that $G$ is $k$-$(S,T)$-edge-connected,
if $G$ has $k$ edge-disjoint $st$-paths for every $(s,t) \in S \times T$. 

\begin{center} 
\fbox{
\begin{minipage}{0.965\textwidth}
\noindent
{\sf $k$-$(S,T)$-Edge-Connected Subgraph} \\
{\em Instance:}
\ A directed graph $G=(V,E)$ with edge-costs $\{c_e:e \in E\}$, disjoint subsets $S,T \subseteq V$, 
and an integer $k$. \\
{\em Objective:}
Find a minimum cost $k$-edge-outconnected from $S$ to $T$ subgraph of $G$.
\end{minipage}
}
\end{center}

In the so called ``standard version'' of the problem we have $E \subseteq S \times T$.
Let the {\sf $\ell$-$(S,T)$-Edge-Connectivity Augmentation} problem be the restriction of 
{\sf $k$-$(S,T)$-Edge-Connected Subgraph} to instances in which $G$ contains a
subgraph $G_0$ of cost zero such that $G_0$ is $\ell$-$(S,T)$-connected in $G_0$,
and we seek to increase at minimum cost the $(S,T)$-connectivity from $\ell=k-1$ to $\ell+1=k$. 

Let us say that two sets $X,Y$ {\em $(S,T)$-cross} if $X \cap Y \cap S, T \setminus (X \cup Y) \neq \emptyset$.
A set-family ${\cal F}$ is {\em $(S,T)$-crossing} 
if $X \cap Y, X \cup Y \in {\cal F}$ for any $X,Y \in {\cal F}$ that $(S,T)$-cross.
Generalizing the algorithm of Fackaroenphol and Laekhanukit \cite{FL} for the 
{\sf $k$-Connected Subgraph} problem, 
Cheriyan and Laekhanukit \cite{ChL} gave an approximation
algorithm with ratio $O(\log \nu({\cal F}))=O(\log n)$ 
for the standard version of the {\sf $k$-$(S,T)$-Connectivity Augmentation} problem.
They also implicitly proved that it is a particular case of the {\sf Set-Family Edge-Cover} problem 
with $V=S \cup T$, $E \subseteq S \times T$, and $(S,T)$-crossing ${\cal F}$.
Our algorithm in Theorem~\ref{t:FL}(i) easily extends to the problem 
of covering an $(S,T)$-crossing family by a minimum-cost edge-set. 
Here we preferred the biset-family setting for simplicity of exposition, and since
the concept of $k$-regularity is not a natural one for $(S,T)$-crossing families.
Furthermore, the case of an $(S,T)$-crossing set family ${\cal F}$ 
is reduced to the case of a crossing biset-family ${\cal F}'$,
where for every set $X \in {\cal F}$ there is a biset $\hat{X}'=(X \cap S, S \cup (X \cap T))$ in ${\cal F}'$;
it is not hard to verify that if ${\cal F}$ is an $(S,T)$-crossing family then ${\cal F}'$ is 
a crossing biset-family, and that an edge from $S$ to $T$ covers a set $X$ if, and only if, it covers $\hat{X}'$. 
Thus from Theorem~\ref{t:FL}(i) we have the following generalization of the result of \cite{ChL}. 

\begin{corollary}
{\sf Set-Family Edge-Cover} with $(S,T)$-crossing ${\cal F}$ and $E \subseteq S \times T$ 
admits a polynomial time algorithm that computes a solution of cost 
$O(\log |S \cup T|) \cdot \tau({\cal F})$. 
\end{corollary}

\section{Proof of Theorem~\ref{t:FL}}

\subsection{Proof of part~(i)}

Recall that we assume that for any edge set $J$ on $V$ and any $u,v \in V$ we are able to compute 
in polynomial time the cores of the biset-family 
${\cal F}(u,v)=\{\hat{S} \in {\cal F}^J: u \in S, v \in V \setminus S^+\}$ 
and also the cores of its co-family, or to determine that ${\cal F}(u,v)$ is empty.
Note that if ${\cal F}$ is crossing, then ${\cal F}(u,v)$ has a unique core, and the 
co-family of ${\cal F}(u,v)$ also has a unique core.

 \begin{lemma} \label{l:poly}
A crossing biset-family ${\cal F}$ has at most $n(n-1)$ cores and they can be computed in polynomial time.
\end{lemma}
\begin{proof}
For every ordered pair of nodes $(u,v) \in V \times V$ we compute the core $\hat{C}_{uv}$ 
of the biset-family ${\cal F}(u,v)=\{\hat{S} \in {\cal F}^J: u \in S, v \in V \setminus S^+\}$, 
if ${\cal F}(u,v)$ is non-empty.
Then ${\cal C}({\cal F})$ consists from the inclusion-minimal members (cores) 
of the biset-family $\{\hat{C}_{uv}: u,v \in V\}$. 
\end{proof}

\begin{definition}
Given a biset-family ${\cal F}$ and a core $\hat{C} \in {\cal C}({\cal F})$ of ${\cal F}$ 
let ${\cal F}(\hat{C})$ denote the family 
of the bisets in ${\cal F}$ that contain $\hat{C}$ and contain no other core of ${\cal F}$. 
\end{definition}

The following statement can be easily verified. 

\begin{claim} \label{c:iterative-merging}
Let ${\cal F}$ be a biset-family and $J$ a set of directed edges on $V$.
If for some $\hat{C} \in {\cal C}({\cal F})$,
$J$ covers ${\cal F}(\hat{C})$ and covers no core distinct from $\hat{C}$, 
then ${\cal C}({\cal F}^J)={\cal C}({\cal F}) \setminus \{\hat{C}\}$ and 
$\nu({\cal F}^J) = \nu({\cal F})-1$.
Furthermore, if ${\cal F}$ is intersection-closed and $J$ covers ${\cal F}(\hat{C})$ 
for every ${\cal F}$-core $\hat{C}$, 
then every ${\cal F}^J$-core contains at least two ${\cal F}$-cores, 
and thus $\nu({\cal F}^J) \leq \nu({\cal F})/2$.
\end{claim}

\begin{lemma} \label{l:HH}
Let $\hat{C}_1,\hat{C}_2 \in {\cal C}({\cal F})$ be distinct cores of a crossing biset-family ${\cal F}$ 
and let $\hat{S}_1 \in {\cal F}(\hat{C}_1)$ and  $\hat{S}_2 \in {\cal F}(\hat{C}_2)$. 
Then $\hat{S}_1,\hat{S}_2$ do not cross. Consequently, no edge covers both $\hat{S}_1,\hat{S}_2$.
\end{lemma}
\begin{proof}
Suppose to the contrary that $\hat{S}_1$ and $\hat{S}_2$ cross.
Then $\hat{S}_1 \cap \hat{S}_2 \in {\cal F}$.
Thus $\hat{S}_1 \cap \hat{S}_2$ contains some ${\cal F}$-core $\hat{C}$.
We cannot have $\hat{C} \neq \hat{C}_1$ as $\hat{C} \subseteq \hat{C}_1$ and $\hat{C}_1$ is a core, and
we cannot have $\hat{C}=\hat{C}_1$      as $\hat{C} \subseteq \hat{C}_2$, $\hat{C}_1 \neq \hat{C}_2$,
and $\hat{C}_2$ is a core. This gives a contradiction. 
The second statement follows from the observation
an edge covers two bisets $\hat{X},\hat{Y}$ simultaneously if, and only if, it goes from 
$X \cap Y$ to $V \setminus (X^+ \cup Y^+)$, and hence $\hat{X},\hat{Y}$ cross.
\end{proof}

\begin{lemma} \label{l:H}
Let $\hat{C} \in {\cal C}({\cal F})$ be a core of a crossing biset-family ${\cal F}$, 
and let $\hat{X}, \hat{Y} \in {\cal F}(\hat{C})$. 
If $\hat{X}, \hat{Y}$ cross then $\hat{X} \cap \hat{Y}, \hat{X} \cup \hat{Y} \in {\cal F}(\hat{C})$.
\end{lemma}
\begin{proof}
Since ${\cal F}$ is crossing, $\hat{X} \cap \hat{Y}, \hat{X} \cup \hat{Y} \in {\cal F}$.
Since $\hat{X} \cap \hat{Y} \subseteq \hat{X} \subseteq \hat{X} \cup \hat{Y}$
and since $\hat{X} \in {\cal F}(\hat{C})$, it follows that $\hat{X} \cap \hat{Y} \in {\cal F}(\hat{C})$
and that $\hat{C} \subseteq \hat{X} \cup \hat{Y}$.
It remains to prove that $\hat{X} \cup \hat{Y}$ contains no core distinct from $\hat{C}$.
Suppose to the contrary that $\hat{X} \cup \hat{Y}$ contains a core $\hat{S}$ distinct from $\hat{C}$.
Since $\hat{X}, \hat{Y} \in {\cal F}(\hat{C})$, none of $\hat{X},\hat{Y}$ contains $\hat{S}$.
This implies that $\hat{S},\hat{X}$ cross or $\hat{S},\hat{Y}$ cross,
so $\hat{S} \cap \hat{X} \in {\cal F}$ or $\hat{S} \cap \hat{Y} \in {\cal F}$.
This contradicts that $\hat{S}$ is a core.
\end{proof}

\begin{lemma} \label{l:intersecting}
Let ${\cal F}$ be a crossing biset-family on $V$ and let $\hat{C} \in {\cal C}({\cal F})$.
Then the co-family ${\cal R}(\hat{C})=\{(V \setminus S^+, V \setminus S): \hat{S} \in {\cal F}(\hat{C})\}$
of ${\cal F}(\hat{C})$ is intersecting, and its cores can be found in polynomial time.
\end{lemma}
\begin{proof}
Let $\hat{X}_0,\hat{Y}_0 \in {\cal R}(\hat{C})$ be the co-bisets of $\hat{X},\hat{Y} \in {\cal F}(\hat{C})$,
respectively. Suppose that $\hat{X}_0,\hat{Y}_0$ intersect.
Then $\hat{X},\hat{Y}$ cross, hence $\hat{X} \cap \hat{Y}, \hat{X} \cup \hat{Y} \in {\cal F}(\hat{C})$,
by Lemma~\ref{l:H}. The co-bisets of $\hat{X} \cap \hat{Y}$ and $\hat{X} \cup \hat{Y}$ are
$\hat{X}_0 \cup \hat{Y}_0$ and $\hat{X}_0 \cap \hat{Y}_0$,
hence $\hat{X}_0 \cup \hat{Y}_0,\hat{X}_0 \cap \hat{Y}_0 \in {\cal R}(\hat{C})$.
This implies that ${\cal R}(\hat{C})$ is an intersecting biset family. 
Now we show how to find the cores of ${\cal R}(\hat{C})$ in polynomial time.
For an ${\cal F}$-core $\hat{S} \neq \hat{C}$ let 
$K_{\hat{S}}=\{uv: u \in S,v \in V \setminus S^+\}$ be the set of all edges 
from $S$ to $V \setminus S^+$. 
Let $K=\bigcup_{\hat{S} \in {\cal C}({\cal F}) \setminus \{\hat{C}\}} K_{\hat{S}}$. 
We claim that ${\cal F}^K={\cal F}(\hat{C})$.
To see this, note that for any $\hat{S} \in {\cal C}({\cal F}) \setminus \{\hat{C}\}$: \\ 
(i) \ $K_{\hat{S}}$ covers all bisets in ${\cal F}$ that contain $\hat{X}$; 
(ii)  $K_{\hat{S}}$ does not cover any biset in ${\cal F}(\hat{C})$, by Lemma~\ref{l:HH}. 
Now choose $u \in C$, and for every $v \in V$ compute the core $\hat{C}_v$ 
of the co-family of ${\cal F}^K(u,v)$. 
The ${\cal R}(\hat{C})$-cores are the inclusion-minimal members of the family $\{\hat{C}_v: v \in V\}$. 
\end{proof}

\begin{corollary} \label{c:S-greedy}
{\sf Biset-family Edge-Cover} with crossing ${\cal F}$
admits a polynomial time algorithm that given a core $\hat{C} \in {\cal C}({\cal F})$ 
computes an ${\cal F}(\hat{C})$-cover $J_{\hat{C}} \subseteq E$ of cost 
$c\left(J_{\hat{C}}\right)=\tau({\cal F}(\hat{C}))$.
Moreover, $\sum_{\hat{C} \in {\cal C}} \tau({\cal F}(\hat{C})) \leq \tau({\cal F})$, and thus
there exists $\hat{C} \in {\cal C}({\cal F})$ such that $c(J_{\hat{C}}) \leq \tau({\cal F})/\nu({\cal F})$.
\end{corollary}
\begin{proof}
By Lemma~\ref{l:intersecting}, the co-family 
${\cal R}(\hat{C})$ of ${\cal F}(\hat{C})$ is intersecting.
Thus, after reversing the edges in $E$, we can apply a standard primal-dual algorithm to compute an 
edge-cover of ${\cal R}(\hat{C})$ of cost $\tau({\cal R}(\hat{C}))=\tau({\cal F}(\hat{C}))$; 
$J_{\hat{C}}$ is the reverse edge set of this cover.
This primal-dual algorithm can be implemented in polynomial time if the cores of ${\cal R}(\hat{C})$
can be found in polynomial time, which is possible by Lemma~\ref{l:intersecting}.
The second statement of the lemma follows from Lemma~\ref{l:HH}.
\end{proof}

Now we finish the proof of part~(i) of Theorem~\ref{t:FL}.
The algorithm start with $J=\emptyset$.
At iteration $i$, it finds $\hat{C}_i \in {\cal C}({\cal F}^J)$ and $J_i \subseteq E \setminus J$ with 
$c(J_i) \leq \tau({\cal F}^J)/\nu({\cal F}^J)$ and 
$\nu\left({\cal F}^{J \cup J_i}\right)=\nu\left({\cal F}^J\right)-1$, and adds $J_i$ to $J$; 
such $J_i$ exists and can be found in polynomial time by 
Lemma~\ref{l:poly},
Claim~\ref{c:iterative-merging},
Lemma~\ref{l:HH}, and
Corollary~\ref{c:S-greedy}.
At each iteration $\nu({\cal F}^J)$ decreases by $1$, by Claim~\ref{c:iterative-merging}.
Thus at the end of iteration $i$ we have $\nu({\cal F}^J)=\nu({\cal F})-i$.
Consequently, $c(J_i) \leq \tau({\cal F}^J)/\nu({\cal F}^J) \leq \tau({\cal F})/(\nu({\cal F})-i)$.
Thus at the end of the algorithm, 
$c(J) \leq \sum_i c(J_i) \leq \tau({\cal F}) \sum_i 1/(\nu({\cal F})-i) = \tau({\cal F}) \cdot H(\nu({\cal F}))$.

\subsection{Proof of part~(ii)} \label{s:main}

The following concept plays a central role in the proof of part~(ii) of Theorem~\ref{t:FL}.

\begin{definition} \label{d:q}
A biset-family ${\cal S}$ is {\em intersection-closed} if $\hat{X} \cap \hat{Y} \in {\cal S}$
for any intersecting $\hat{X},\hat{Y} \in {\cal S}$.
An intersection-closed biset-family ${\cal S}$ is {\em $q$-semi-intersecting} if 
$|S| \leq q$ for all $\hat{S} \in {\cal S}$, and if 
$\hat{X} \cup \hat{Y} \in {\cal S}$ for any intersecting $\hat{X},\hat{Y} \in {\cal S}$ with $|X \cup Y| \leq q$.
\end{definition}


\begin{fact} \label{f:q}
If a biset-family ${\cal F}$ is $k$-regular, then the subfamily 
${\cal S}=\{\hat{S} \in {\cal F}: |S| \leq q\}$ of ${\cal F}$ is $q$-semi-intersecting 
(and in particular, is intersection closed) for any $q \leq (n-k)/2$.
\end{fact}

The following statement is straightforward.

\begin{lemma} \label{l:HC}
Let ${\cal F}$ be an intersection-closed biset-family.
If $\hat{C} \in {\cal C}({\cal F})$ and $\hat{S} \in {\cal F}$ intersect, then $\hat{C} \subseteq \hat{S}$.
Thus the ${\cal F}$-cores are pairwise disjoint.
\end{lemma}

Let ${\cal S}=\{\hat{S} \in {\cal F}: |S| \leq (n-k)/2\}$.
We will give an algorithm that computes an ${\cal S}$-cover of cost at most
$\tau({\cal F}) \cdot \min\left\{1+H\left(\left\lfloor 2n/(n-k+2) \right\rfloor \right),
\left\lfloor \log_2 \left\lfloor (n-k+2)/2 \right\rfloor \right\rfloor \right\}$.
To cover the entire ${\cal F}$, we apply this algorithm twice: 
once on ${\cal F},E$ and once on the ''reversed'' instance with the biset-family being the co-family 
$\{(V-S^+,V-S):(S,S^+) \in {\cal F}\}$ of ${\cal F}$ and with the reverse edge-set $\{uv: vu \in E\}$ of $E$;
after a solution $J'$ to the reversed instance is computed, we return the reversed edge-set of $J'$.
The union of the two partial solutions computed covers the entire ${\cal F}$,
and has cost as stated in part~(ii) of Theorem~\ref{t:FL}.

The algorithm that computes an ${\cal S}$-cover of cost at most 
$\tau({\cal F}) \cdot \lfloor \log_2 \lfloor (n-k+2)/2 \rfloor \rfloor$ is as follows.
Start with $J=\emptyset$, and iteratively, until $\nu({\cal S}^J)=0$, find and add to $J$
a cover $\bigcup_{\hat{C} \in {\cal C}({\cal F})} J_C$ of cost $\leq \tau({\cal F}^J)$
of all families ${\cal F}(\hat{C})$ of ${\cal S}^J$-cores, as in Corollary~\ref{c:S-greedy}.
By Lemma \ref{l:HC} and Claim~\ref{c:iterative-merging}, 
after step $i$ we have $2^i \leq |C| \leq \frac{n-k}{2}$ for every ${\cal S}^J$-core $\hat{C}$.
Hence the number of iterations is at most $\lfloor \log_2 \lfloor (n-k+2)/2 \rfloor \rfloor$. 
As at every iteration we add to $J$ an edge set of cost $\leq \tau({\cal F})$, the total
cost of the ${\cal S}$-cover computed is $\tau({\cal F}) \cdot \lfloor \log_2 \lfloor (n-k+2)/2 \rfloor \rfloor$.

In the next section we will prove the following theorem, which is our main technical result,
and which we believe is of independent interest.

\begin{theorem} \label{t:q-int}
{\sf Biset-Family Edge-Cover} with $q$-semi-intersecting biset-family ${\cal S}$ admits a polynomial time algorithm 
that computes an edge-set $J \subseteq E$ such that $\nu({\cal S}^J) \leq \left\lfloor n/(q+1) \right\rfloor$ 
and $c(J) \leq \tau({\cal S})$. 
\end{theorem}

Note that for $q=n$ any intersecting biset-family is $q$-semi-intersecting, hence 
$q$-semi-intersecting biset-families generalize intersecting biset-families, 
and then the algorithm in Theorem~\ref{t:q-int} computes an optimal ${\cal S}$-cover of cost $\tau({\cal S})$.

Part~(ii) of Theorem~\ref{t:FL} follows from Theorem~\ref{t:q-int} and part~(i).
Compute an edge set $J$ as in Theorem~\ref{t:q-int} and then compute 
a cover $F$ of the residual family ${\cal S}^J$ using the algorithm from part~(i). 
The cost of the ${\cal S}$-cover $J+F$ computed is bounded by
$$
c(J+F)=c(J)+c(F) \leq
\tau({\cal S}) + H(\nu({\cal S}^J)) \cdot \tau({\cal F}) \leq
\tau({\cal F}) \cdot (1+H(\lfloor 2n/(n-\ell) \rfloor)) \ .$$

\subsection{Proof of part~(iii)}

Note that the algorithm from part~(ii), if applied on an arbitrary crossing $k$-regular 
biset-family ${\cal F}$ with $\gamma({\cal F}) \leq k$, 
returns an edge set $J$ of costs $c(J)=\tau({\cal F}) \cdot O\left(\log \frac{n}{n-k}\right)$,
such that the residual family ${\cal F}^J$ has the following property: 
the size of the inner part of every biset
in ${\cal F}^J$ or in the co-family of ${\cal F}^J$ is larger than $q=\frac{n-k}{2}$.
Consequently, to prove part~(iii) of Theorem~\ref{t:FL} it is sufficient to prove
the following.

\begin{lemma} \label{l:q}
{\sf Biset-Family Edge-Cover} with crossing ${\cal F}$
admits a polynomial time algorithm that computes an edge-cover 
of ${\cal F}$ of cost at most 
$\tau({\cal F}) \cdot (n/q) \cdot H\left(\left\lfloor n/q \right\rfloor\right)$,
provided that $|S| \geq q$ holds
for every biset $\hat{S}$ that belongs to ${\cal F}$ or to the co-family of ${\cal F}$.
\end{lemma}

In the rest of this section we prove Lemma~\ref{l:q}. The key observation is the following.

\begin{lemma} \label{l:lov}
Let ${\cal F}$ be a crossing biset-family on $V$ with $|S| \geq q$
for every biset $\hat{S}$ in ${\cal F}$ or in the co-family of ${\cal F}$.
Then there exists $T \subseteq V$ of size 
$|T| \leq (n/q) \cdot H\left(\left\lfloor n/q \right\rfloor\right)$ 
such that $T \cap S \neq \emptyset$ for every $\hat{S} \in {\cal F}$.
\end{lemma}
\begin{proof}
Consider the hypergraph ${\cal H}=\{C: \hat{C} \in {\cal C}({\cal F})\}$
of the inner parts of the cores of ${\cal F}$.
We combine the following two observations.  
\begin{itemize}
\item[(i)]
The function $t$ on $V$ defined by $t(v)=\frac{1}{q}$ for all $v \in V$ is a fractional 
hitting-set of ${\cal H}$ (namely $\sum_{v \in S} t(v) \geq 1$ for all $C \in {\cal H}$)
of value $\sum_{v \in V} t(v) = n/q$. 
\item[(ii)]
The maximum degree in the hypergraph ${\cal H}$ is at most $\lfloor n/q \rfloor$.
\end{itemize}
Given observations (i) and (ii), the greedy algorithm 
computes a subset $T \subseteq V$ as stated.
Observation~(i) follows from the assumption that $|S| \geq q$ for all $\hat{S} \in {\cal F}$.
We prove (ii). Since ${\cal F}$ is crossing, the members of ${\cal C}({\cal F})$ 
are pairwise non-crossing. Thus if ${\cal C} \subseteq {\cal C}({\cal F})$ is a set of cores
which inner parts contain the same element $v \in V$, then the sets in 
$\{V \setminus C^+:\hat{C} \in {\cal C}\}$ are pairwise disjoint.
As each of these sets is an inner part of a biset in the co-family of ${\cal F}$,
the number of such sets is at most $\lfloor n/q \rfloor$. Observation~(ii) follows.
\end{proof}

Lemma~\ref{l:q} easily follows from Lemma~\ref{l:lov}.
Note that if ${\cal F}$ is crossing, then 
for every $s \in V$ the co-family of the biset-family 
${\cal F}_s=\{\hat{S} \in {\cal F}:v \in S\}$ is intersecting.
Thus given an instance of {\sf Biset-Family Edge-Cover} and $s \in V$,
we can compute in polynomial time an edge-cover $J_s$ of ${\cal F}_s$ of cost
$c(J_s) \leq \tau({\cal F}_s) \leq \tau({\cal F})$.
Now let $T$ be as in Lemma~\ref{l:lov}. For every $s \in T$ we compute an edge-cover 
$J_s$ of ${\cal F}_s$ as above, and return $J=\cup_{s \in T}J_s$. 
This concludes the proof of Lemma~\ref{l:q} and thus also the proof of part~(iii)
of Theorem~\ref{t:FL} is now complete.

\section{Proof of Theorem~\ref{t:q-int}} \label{s:q-int}

\begin{definition}
For biset-families ${\cal S}$ and ${\cal U}$ on $V$ let 
$${\cal S}\left[{\cal U}\right]=
\{\hat{S} \in {\cal S}: \hat{S} \subseteq \hat{U} \mbox{ for some } \hat{U} \in {\cal U}\} \ .$$
A biset-family ${\cal S}$ is {\em weakly-intersecting} if ${\cal S}\left[\{\hat{U}\}\right]$ is 
an intersecting biset-family for every $\hat{U} \in {\cal S}$.
\end{definition}

Clearly, any $q$-semi-intersecting biset-family is weakly-intersecting.
Note that if ${\cal U} \subseteq {\cal S}$ 
and if the members of ${\cal U}$ are pairwise disjoint, 
then ${\cal S}[{\cal U}]$ is an intersecting biset-family, if ${\cal S}$ is weakly-intersecting.
We will prove the following refinement of Theorem~\ref{t:q-int}.

\begin{lemma} \label{l:P12}
{\sf Biset-Family Edge-Cover} with weakly-intersecting biset-family ${\cal S}$ admits a polynomial time algorithm 
that computes a sub-family ${\cal U} \subseteq {\cal S}$ of pairwise disjoint bisets  
and an edge-set $J \subseteq E$ such that the following two properties hold.
\begin{description}
\item[Property 1:] 
For any ${\cal S}^J$-core $\hat{C}$ the following holds:
\begin{itemize}
\item[{\em (i)}]
If $\hat{U} \in {\cal U}$ and $\hat{C}$ intersect then $\hat{U}$ intersects 
no ${\cal S}^J$-core distinct from $\hat{C}$. 
\item[{\em (ii)}]
The union $\hat{B}_C$ of $\hat{C}$ and the bisets in ${\cal U}$ intersecting with $\hat{C}$
is not in ${\cal S}$.
\end{itemize}
\item[Property 2:] 
$J$ is an optimal cover of ${\cal S}[{\cal U}]$, thus
$c(J)=\tau({\cal S}[{\cal U}]) \leq \tau({\cal S})$. 
\end{description}
\end{lemma}

To see that Lemma~\ref{l:P12} implies Theorem~\ref{t:q-int},
it is sufficient to show that if ${\cal S}$ is $q$-semi-intersecting, then 
Property~1 implies $\nu({\cal S}^J) \leq \left\lfloor n/(q+1) \right\rfloor$.
Let $B_C$ be the inner part of $\hat{B}_C$.
Note that the sets $B_C$ are pairwise disjoint;
this is since the members of ${\cal U}$ are pairwise disjoint, the ${\cal S}^J$-cores are pairwise disjoint, 
and since by Property~1(i) every $\hat{U} \in {\cal U}$ intersects at most one ${\cal S}^J$-core.
Now observe that if ${\cal S}$ is $q$-semi-intersecting then $|B_C| \geq q+1$, 
since $\hat{B}_C \notin {\cal S}$, by Property~1(ii).
Consequently, $\nu({\cal S}^J) \leq \left\lfloor n/(q+1) \right\rfloor$,
and Lemma~\ref{l:P12} implies Theorem~\ref{t:q-int}.

In the rest of this section we prove Lemma~\ref{l:P12}.
The algorithm is a variation of a standard primal-dual algorithm
for covering an {\em intersecting} biset-family, and it only needs that the 
${\cal S}^J$-cores can be computed in polynomial time. 
The dual LP of the LP-relaxation {\bf (P)} from the Introduction is:
\[ \displaystyle
\begin{array} {llll} \mbox{\bf (D)}
& \max          & \ \ \ \displaystyle \sum_{\hat{S} \in {\cal S}} y_S  \\
& \ \mbox{s.t.} & \displaystyle \sum_{\hat{S} \in {\cal S}: e \in \delta(\hat{S})} y_{\hat{S}} \leq c_e 
                  \ \ \ \ \ \forall e \in E \\
&               & \ \ \ \ y_{\hat{S}} \geq 0 \ \ \ \ \ \ \ \ \ \ \ \ \ \ \ \forall \hat{S} \in {\cal S}
\end{array}
\]
Given a partial solution $y$ to {\bf (D)}, an edge $e \in E$ is {\em tight}
if the inequality in {\bf (D)} that corresponds to $e$ holds with equality.
The algorithm produces an edge set $J \subseteq E$, a sub-family ${\cal U} \subseteq {\cal S}$ of ${\cal S}$,
and a dual solution $y$ to {\bf (D)}, such that the following holds: 
$J$ covers ${\cal S}[{\cal U}]$, $y$ is a feasible solution to {\bf (D)},
and (the characteristic vector of) $J$ and $y$ satisfy the complementary slackness conditions: 

\noindent
$\displaystyle \mbox{\em Primal Complementary Slackness Conditions:} \ \   e \in J \Longrightarrow \mbox{$e$ is tight}$; \\
$\displaystyle \mbox{\em Dual Complementary Slackness Conditions:} \ \ \  y_S>0 \Longrightarrow |\delta_J(\hat{S})|=1$.

{\bf Phase~1} starts with $J=\emptyset$ and ${\cal U}=\emptyset$ and applies a sequence of iterations.
At each iteration we choose some ${\cal F}^J$-core $\hat{C}$ and do the following:
\begin{enumerate}
\item
Add $\hat{C}$ to ${\cal U}$ and exclude from ${\cal U}$ the bisets contained in $\hat{C}$.
\item
Raise (possibly by zero) the dual variable corresponding to $\hat{C}$, 
until some edge $e \in E \setminus J$ covering $\hat{C}$ becomes tight, and add $e$ to $J$.
\end{enumerate}

Phase~1 terminates when $\nu({\cal S}^J)=0$, namely, when $J$ covers ${\cal S}$. 

\vspace{0.2cm}

{\bf Phase~2} applies on $J$ ``reverse delete'' like the family ${\cal S}[{\cal U}]$ 
is the one we want to cover, which means the following.
Let $J=\{e_1, \ldots, e_j\}$, where $e_{i+1}$ was added after $e_i$. 
For $i=j$ downto $1$, we delete $e_i$ from $J$ if $J-\{e_i\}$ still covers 
the family ${\cal S}[{\cal U}]$.
This can be implemented in polynomial time as follows. 
When an edge $e \in J$ is checked for deletion,
${\cal S}[{\cal U}]$ is covered by $J \setminus \{e\}$ if, and only if, 
no $\hat{U} \in {\cal U}$ contains an ${\cal S}^{J \setminus \{e\}}$-core.   
At the end of the algorithm, $J$ is output.

\vspace{0.2cm}

Summarizing, the algorithm is a variation of a standard primal-dual algorithm 
for intersecting biset-families, with the following changes.
\begin{enumerate}
\item
Unlike a standard primal-dual algorithm in which {\em all} the dual variables corresponding to 
cores are raised uniformly, we raise the dual variable of only {\em one} core.
\item
The algorithm maintains a biset-family ${\cal U} \subseteq {\cal S}$.
In each iteration, we add to ${\cal U}$ the corresponding tight ${\cal F}^J$-core $\hat{C}$,
and exclude from ${\cal U}$ all the bisets contained in $\hat{C}$, if any.
\item
While at Phase~1 the algorithm intends to cover the entire biset-family ${\cal S}$,
Phase~2 (reverse-delete) is applied like the family ${\cal S}[{\cal U}]$ is the one that we want to cover.
\end{enumerate}

Using Lemma~\ref{l:poly}, it is easy to see that the algorithm can be implemented in polynomial time,
under the oracles assumed. Using Lemma~\ref{l:HC} it is also easy to see the following.

\begin{claim} \label{c:feasible}
During the algorithm the following holds:
the members of ${\cal U}$ are pairwise disjoint,
every edge in $J$ has its tail in the inner part of some $\hat{U} \in {\cal U}$, and
$|\delta_{J}(\hat{U})|=1$ for every $\hat{U} \in {\cal U}$.
\end{claim}

Retrospectively, it turns out that our algorithm coincides with some run of an almost 
standard primal-dual algorithm that intends to cover the intersecting biset-family ${\cal S}[{\cal U}]$. 
The difference is in item 2 above, as we raise the dual variable of only one core;
it is also possible to raise all the dual variables corresponding to cores,
but this makes the analysis more complicated.
We will show that this modified algorithm still computes an optimal solution.
This will ensure Property~2 in Lemma~\ref{l:P12}; we give a formal proof of Property~2 after
proving that Property~1 holds for $J,{\cal U}$ at the end of the algorithm.

The following statement is easily verified, see Figure~\ref{f:sets}.

\begin {figure} 
\centering 
\epsfbox{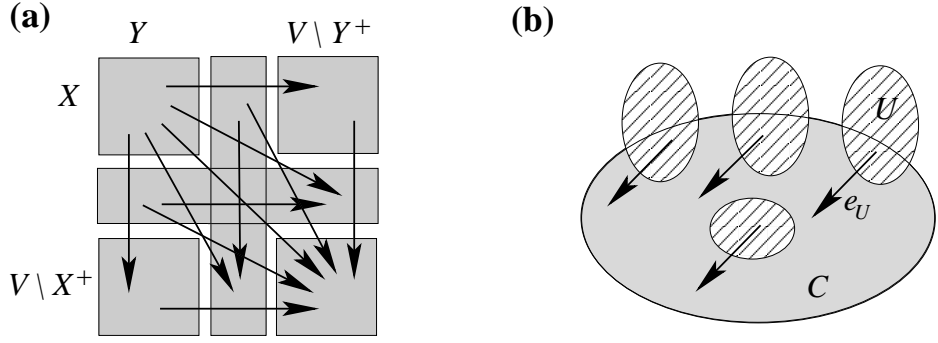}
   \caption {(a) All types of edges that can cover $\hat{X} \cap \hat{Y}$ or $\hat{X} \cup \hat{Y}$. 
             (b) An ${\cal S}^J$-core $C$ and the members of ${\cal U}$ intersecting $C$.
            }
   \label{f:sets}
\end {figure}

\begin{fact} \label{f:XY}
Let $\hat{X},\hat{Y}$ be bisets and let $e$ be an edge. 
\begin{itemize}
\item[{\em (i)}]
If $e$ covers $\hat{X} \cap \hat{Y}$ or $\hat{X} \cup \hat{Y}$ then $e$ covers $\hat{X}$ or $\hat{Y}$.
\item[{\em (ii)}]
If $e$ covers $\hat{X} \cup \hat{Y}$ and has tail in $X$ then $e$ covers $\hat{X}$. 
\item[{\em (iii)}]
If $e$ covers both $\hat{X} \cap \hat{Y}$ and $\hat{X} \cup \hat{Y}$ then $e$ covers both $\hat{X}$ and $\hat{Y}$.
\end{itemize}
\end{fact}

Let $F$ be the set of edges stored in $J$ at the end of Phase~1.
Then $F$ covers ${\cal S}$. Note that every edge in $F$ 
has its tail in the inner part of some $\hat{U} \in {\cal U}$, and that 
$|\delta_F(\hat{U})|=1$ for every $\hat{U} \in {\cal U}$.
Thus the following statement implies that Property~1 holds for $J$ after Phase~2, 
even if ${\cal S}$ is only intersection-closed.

\begin{lemma}
Let $F$ be an edge-cover of an intersection-closed biset-family ${\cal S}$ and let
${\cal U} \subseteq {\cal S}$ be a sub-family of ${\cal S}$ of pairwise disjoint bisets 
such that every edge in $F$ has its tail in the inner part of some $\hat{U} \in {\cal U}$ and
such that $\delta_F(\hat{U})=\{e_U\}$ for every $\hat{U} \in {\cal U}$.
Let $J \subseteq F$ be a cover of ${\cal S}[{\cal U}]$ and let $\hat{C}$ be an ${\cal S}^J$-core.
Then the following holds (see Figure~\ref{f:sets}(b)):
\begin{itemize}
\item[{\em (i)}]
If $\hat{U} \in {\cal U}$ and $\hat{C}$ intersect then $\delta_J(\hat{C} \cap \hat{U})=\{e_U\}$;
thus $\hat{U}$ intersects no ${\cal S}^J$-core distinct from $\hat{C}$.
\item[{\em (ii)}]
$\delta_F(\hat{B}_C)=\emptyset$, where $\hat{B}_C$ is the union of $\hat{C}$ 
and the bisets in ${\cal U}$ intersecting with $\hat{C}$;
thus $\hat{B}_C \notin {\cal S}$.
\end{itemize}
\end{lemma}
\begin{proof}
We prove (i). Let $\hat{U} \in {\cal U}$ and $\hat{C}$ intersect.  
Note that $\hat{C} \cap \hat{U} \in {\cal S}[{\cal U}]$, since ${\cal S}$ is intersection-closed.
Hence $\delta_J(\hat{C} \cap \hat{U}) \neq \emptyset$.
Let $e \in \delta_J(\hat{C} \cap \hat{U})$.
Then $e$ covers $\hat{C}$ or $\hat{U}$, by Fact~\ref{f:XY}(i).
But $e$ does not cover $\hat{C}$, since $e \in J$, and since $J$ does not 
cover $\hat{C}$. Hence $e$ covers $\hat{U}$. 
Consequently, $e=e_U$ for any $e \in \delta_J(\hat{C} \cap \hat{U})$,
hence $\delta_J(\hat{C} \cap \hat{U})=\{e_U\}$.
This implies that the tail of $e_U$ is in $C$, and it cannot be in the inner part of 
any other ${\cal S}^J$-core, since the ${\cal S}^J$-cores are pairwise disjoint.
Thus $U$ intersects no ${\cal S}^J$-core distinct from $C$.

We prove (ii). 
Suppose to the contrary that there is $e \in \delta_F(\hat{B}_C)$.
Let $\hat{U} \in {\cal U}$ be the biset whose inner part contains the tail of $e$.
Note that $U \cap C \neq \emptyset$.
Let $\hat{X}=\hat{C} \cup \hat{U}$ and 
let $\hat{Y}$ be the union of $\hat{C}$ and the bisets in ${\cal U} \setminus \{\hat{U}\}$
that intersect $\hat{C}$. Note that $\hat{B}_C=\hat{X} \cup \hat{Y}$.
Hence the tail of $e$ lies in $X$ and $e$ covers $\hat{X} \cup \hat{Y}$.
Thus $e$ covers $\hat{C} \cup \hat{U}$, by Fact~\ref{f:XY}(ii).
Applying the same argument on the bisets $\hat{U},\hat{Y}$ we obtain that $e$ covers $\hat{U}$.
Hence $e=e_U$, as $e_U$ is the only edge in $F$ that covers $\hat{U}$.
By (i), $e_U$ covers $\hat{C} \cap \hat{U}$.
Consequently, $e_U$ covers both $\hat{C} \cup \hat{U}$ and $\hat{C} \cap \hat{U}$,
and hence $e_U$ covers both $\hat{C}$ and $\hat{U}$, by Fact~\ref{f:XY}(iii).
However $e_U \in J$, contradicting that $\hat{C} \in {\cal S}^J$.
\end{proof}

For Property~2 it is sufficient to prove the following.

\begin{lemma} 
If ${\cal S}$ is weakly-intersecting, then at the end of the algorithm the following holds: 
$J$ co\-vers the biset-family ${\cal S}[{\cal U}]$,
$y$ is a feasible solution to {\em \bf (D)} for ${\cal S}[{\cal U}]$, 
and $J,y$ satisfy the complementary slackness conditions;
hence $J$ and $y$ are optimal solutions.
\end{lemma}
\begin{proof}
It is clear that $J$ covers ${\cal S}[{\cal U}]$, 
and that during the algorithm $y$ remains a feasible solution to {\bf (D)}.
Since only tight edges enter $J$, after Phase~1 the {\em Primal Complementary Slackness Conditions} hold for $J$.
So, the only part that requires proof is that after Phase~2, the 
{\em Dual Complementary Slackness Conditions} hold for $J$ and $y$.

\vspace{0.1cm}

\noindent
{\em  Claim:}
{\em Consider an arbitrary $\hat{S} \in {\cal S}$ with $y_{\hat{S}}>0$ and an edge $e \in \delta_J(\hat{S})$.
Then there exists $\hat{W} \in {\cal S}[{\cal U}]$ such that $\delta_J(\hat{W})=\{e\}$ and 
$\hat{S} \subseteq \hat{W} \subseteq \hat{U}$ for some $\hat{U} \in {\cal U}$.}

\vspace{0.1cm}

\noindent
{\em Proof:}
Such $\hat{W}$ can be chosen as any member of ${\cal S}[{\cal U}]$
which becomes uncovered if we delete (instead of keeping)
$e$ at the reverse delete step when $e$ was considered for deletion;
note that the algorithm decided to keep $e$, hence such $\hat{W}$ exists.
Moreover, since the edges were deleted in the reverse order,
$\hat{W} \in {\cal S}[{\cal U}]^{J \setminus\{e\}}$. 
Obviously, $\delta_J(\hat{W})=\{e\}$ and $\hat{S}$ and $\hat{W}$ intersect.
Finally, to see that $\hat{S} \subseteq \hat{W}$ note that: 
(i) at any iteration before $e$ was added, $\hat{W}$ was uncovered;
(ii) since $y_{\hat{S}}>0$, there was an iteration before $e$ was added 
at which $\hat{S}$ was an ${\cal S}^J$-core. Hence $\hat{S} \subseteq \hat{W}$, by Lemma~\ref{l:HC}. 
\hfill $\Box$

Now assume to the contrary that there is $\hat{S} \in {\cal S}[{\cal U}]$ with $y_{\hat{S}}>0$ 
such that there are $e',e'' \in \delta_J(\hat{S})$, $e' \neq e''$.
Let $\hat{U} \in {\cal U}$ such that $\hat{S} \subseteq \hat{U}$.
Let $\hat{W}',\hat{W}''$ be bisets for $e',e''$ as in the Claim above. 
so $\delta_J(\hat{W}')=\{e'\}$, $\delta_J(\hat{W}'')=\{e''\}$ and 
$\hat{S} \subseteq \hat{W}' \cap \hat{W}'' \subseteq \hat{U}$.
In particular, $\hat{W}',\hat{W''}$ intersect and $\hat{W}',\hat{W}'' \subseteq \hat{U}$.
Thus $\hat{W}' \cup \hat{W}'' \in {\cal S}[{\cal U}]$, since $\hat{U} \in {\cal S}$ 
and since ${\cal S}$ is weakly intersecting.
Consequently, there is an edge $e \in \delta_J(\hat{W}' \cup \hat{W}'')$. 
This implies that $e \in \delta_J(\hat{W}')$ or $e \in \delta_J(\hat{W}'')$, by Fact~\ref{f:XY}(i).
Consequently, $e=e'$ or $e=e''$.
Since the tail of each one of $e',e''$ is in $S \subseteq W' \cap W''$, so is the tail of $e$.
The head of $e$ is in $V \setminus (W'^+ \cup W''^+)$, since $e$ covers $\hat{W}' \cup \hat{W}''$.
We conclude that $e \in \delta_J(\hat{W}') \cap \delta_J(\hat{W}'')$. This is a contradiction since 
$\delta_J(\hat{W}') = \{e'\}$, $\delta_J(\hat{W}'') = \{e''\}$, and $e' \neq e''$. 
\end{proof}

The proof of Lemma~\ref{l:P12}, and thus also of Theorem \ref{t:q-int} is complete.

\section{Concluding remarks and open problems}
The main open question is whether the {\sf $k$-Connectivity Augmentation} 
problem admits a constant ratio approximation algorithm also for large values of $k$,
say $k=n-\sqrt{n}$. In this context we mention four papers.
In \cite{LN} is is shown that for values of $k$ close to $n$, 
the approximability of the {\sf $k$-Connectivity Augmentation} problem is the same 
for directed and undirected graphs, up to a factor of $2$. 
Therefore, one should not expect to obtain a constant ratio for undirected graphs only. 
On the positive side, Frank and Jordan \cite{FJ} showed that for directed graphs, {\sf $k$-Connected Subgraph} 
can be solved in polynomial time when the input graph is complete and the costs are in $\{0,1\}$.
For arbitrary costs however, there are two negative results.
In \cite{RWe} Ravi and Williamson gave an example showing that the approximation
ratio of a standard primal-dual algorithm that intends to edge-cover 
the biset-family ${\cal S}$ as in Fact~\ref{f:q} has approximation ratio $\Omega(k)$.
This does not exclude that some other variation of the primal-dual algorithm, that relies on concepts from \cite{FJ}
has a constant ratio for crossing biset-families.
However recently, Aazami, Cheriyan, and Laekhanukit \cite{ACL} showed that the standard iterative 
rounding method that is based on a standard LP-relaxation for {\sf $k$-Connectivity Augmentation}
(thus intends to edge-cover a crossing $k$-regular biset-family) has approximation ratio $\Omega(\sqrt{k})$. 


\end{document}